\documentclass[]{article}
\usepackage{arxiv}
\usepackage{booktabs} 
\usepackage[ruled]{algorithm2e} 

\SetAlFnt{\small}
\SetAlCapFnt{\small}
\SetAlCapNameFnt{\small}
\SetAlCapHSkip{0pt}
\IncMargin{-\parindent}

\usepackage{optidef}
\usepackage{bbm}
\usepackage{tikz}
\usetikzlibrary{decorations.pathmorphing}
\usetikzlibrary{positioning}

\usepackage{graphicx}
\usepackage{textcomp}
\usepackage{xcolor}
\usepackage{url}

\DeclareMathOperator{\forbidden}{forbidden}

\begin{document}
\title{A Lattice Linear Predicate Parallel Algorithm for the Dynamic Programming Problems}
\author{Vijay K. Garg
   \\
    Department of Electrical and Computer  Engineering\\
  University of Texas at Austin,\\
  Austin, TX 78712\\
  \texttt{garg@ece.utexas.edu} \\
}

%

\newtheorem{theorem}{Theorem}
\newtheorem{lemma}{Lemma}
\newtheorem{corollary}{Corollary}
\newtheorem{e1}{Example}
\newtheorem{definition}{Definition}
\def\IEEEQED{\mbox{\rule[0pt]{1.3ex}{1.3ex}}} 
\newcommand{\ep}{\hspace*{\fill}~\IEEEQED}
\newenvironment{proof}{{\bf Proof:}}{\ep}

\newcommand{\remove}[1]{}
\newcommand{\h}{\hspace*{0.2in}}
\newcommand{\Ra}{\Rightarrow}
\newcommand{\ra}{\rightarrow}
\newcommand{\CC}{{L}}
\newcommand{\CB}{{L_B}}
\newcommand{\C}{G}
\newcommand{\RR}{\mathbf{R}}






\bibliographystyle{unsrt}


\maketitle

\begin{abstract}
It has been shown that the parallel Lattice Linear Predicate (LLP) algorithm solves many combinatorial optimization problems such as the shortest path problem, the stable marriage problem and the market clearing price problem.
In this paper, we give the parallel LLP algorithm for many dynamic programming problems. In particular, we show that the LLP algorithm solves the longest subsequence problem, the optimal binary search tree problem, and the knapsack problem. 
Furthermore, the algorithm can be used to solve the constrained versions of these problems so long as the constraints are lattice linear. The parallel LLP algorithm requires only read-write atomicity and no higher-level atomic instructions. 
 \end{abstract}

\section{Introduction}

It has been shown that the  Lattice Linear Predicate (LLP) algorithm solves many combinatorial optimization problems such as the shortest path problem, the stable marriage problem and the market clearing price problem \cite{DBLP:conf/spaa/Garg20}. 
In this paper, we show that many problems that can be solved using dynamic programming \cite{bellman1952theory} can also be solved in {\em parallel} using the LLP algorithm.
Dynamic programming is applicable to problems where it is easy to set up a recurrence relation
such that the solution of the problem can be derived from the solutions of problems with smaller sizes.
One can solve the problem using recursion; however, recursion may result in many duplicate computations.
By using memoization, we can avoid recomputing previously computed values. We assume that the problem is solved using dynamic programming with
such bottom-up approach in this paper. 

The LLP algorithm views solving a problem as searching for an element in a finite distributive lattice  \cite{Birk3, davey, Gar:2015:bk} that satisfies a given predicate $B$. The predicate
is required to be closed under the operation of meet (or, equivalently lattice-linear, defined in Section \ref{sec:back}). For all the problems considered in the paper, 
the longest subsequence problem, the optimal binary search tree problem and the Knapsack problem, the predicate is indeed closed under meets.
Any finite distributive lattice can be equivalently characterized by a finite poset of its join-irreducibles from Birkhoff's theorem  \cite{Birk3, davey}. 
The LLP algorithm works on the finite poset in parallel to find the least element in the distributive lattice that satisfies the given predicate.
It starts with the bottom element of the lattice and marches towards the top element of the lattice in a {\em parallel} fashion by advancing
on any chain of the poset for which the current element is {\em forbidden}. 

There are also some key differences between dynamic programming (the bottom-up approach) and the LLP algorithm. The usual dynamic programming problem seeks a structure that
minimizes (or maximizes) some scalar. For example, the longest subsequence problem asks for the subsequence in an array $A[1..n]$ that maximizes the sum.
In contrast, the LLP algorithm seeks to minimize or maximize 
a {\em vector}. In the longest subsequence problem with the LLP approach, we are interested in the longest subsequence in the array $A[1..i]$ for each $i \leq n$ that ends at index $i$. Thus, instead of 
asking for a scalar, we ask for the vector of size $n$. 
We get an array $G[1..n]$ and the solution to the original problem is just the maximum value in the array $G$.
Similarly, the optimal binary search tree problem \cite{knuth1971optimum} asks for the construction of an optimal binary search tree 
on  $n$ symbols such that each symbol $i$ has probability $p_i$ of being searched. Our goal is to find the binary search tree
that minimizes the expected cost of search in the tree. 
The LLP problem seeks the optimal binary search
tree for all ranges $i \ldots j$ instead of just one range $1..n$. 
Finally, the knapsack problem \cite{horowitz1974computing,ibarra1975fast}
asks for the maximum valued subset of items that can be fit in a knapsack such that the profit is maximized and the total weight of the knapsack is at most $W$. The LLP problem seeks the maximum profit obtained by
choosing items from $1..i$ and the total weight from $1..W$.
In all these problems, traditionally we are seeking a single structure that optimizes a single scalar; whereas the LLP algorithm asks for a vector.
It turns out that that in asking for an optimal {\em vector} instead of an optimal {\em scalar}, we do not lose much since the existing solutions also
 end up finding the optimal solutions for the subproblems. The LLP algorithm returns a vector $G$ such that $G[i]$ is optimal for $i$.
 
 The second difference between dynamic programming and the LLP algorithm is in terms of parallelism. 
 The dynamic programming solution does not explicitly refers to parallelism in the problem. The LLP algorithm
 has an explicit notion of parallelism. The solution uses an array $G$ for all problems and the algorithm requires the components of $G$ to be advanced
 whenever they are found to be {\em forbidden}. If $G[i]$ is forbidden for multiple values of $i$, then $G[i]$ can be advanced for all those values in parallel.
 
 The third difference between dynamic programming and the LLP algorithm is in terms of synchronization required during parallel execution of the algorithm. In case of dynamic programming, if the recursive formulas are evaluated in parallel it is assumed that the values used are correct. In contrast, 
 suppose that we check for $G[i]$ and $G[j]$ to be forbidden concurrently such that $G[i]$ ends up using an old value of
 $G[j]$, the LLP algorithm is still correct. 
 The only requirement we have for parallelism is that when $G[i]$ uses a value of $G[j]$, it should either be the most recent value
 of $G[j]$ or some prior value. A processor that is responsible for keeping $G[i]$ may get old value from $G[j]$ in a parallel setting when it gets this value from a cache. In a message passing system, it may get the old value of $G[j]$ if the message to update $G[j]$ has not yet arrived at the processor with $G[i]$. Thus, LLP algorithms are naturally parallel with little synchronization overhead.
 
 The fourth difference between dynamic programming and the LLP algorithm is that we can use the LLP algorithm to solve a constrained version
 of the problem, so long as the constraint itself is lattice-linear. Suppose that we are interested in the longest subsequence such that successive elements differ by at least $2$.
 It can be (easily) shown that this constraint is lattice-linear. Hence, the LLP algorithm is applicable because we are searching for an element that satisfies a conjunction 
 of two lattice-linear predicates. Since the set of lattice-linear predicates is closed under conjunction, the resulting predicate is also lattice-linear and
 the LLP algorithm is applicable. Similarly, the predicate that the symbol $i$ is not a parent of symbol $j$ is lattice-linear and the constrained optimal binary search tree
 algorithm returns the optimal tree that satisfies the given constraint. In the Knapsack problem, it is easy to solve the problem with the additional constraint that if 
 the item $x$ is included in the Knapsack, then the item $y$ is also included.

We note here that our goal is not to improve the time or work complexity of the algorithms, but to provide a single parallel algorithm that solve all of these problems and their constrained versions. Furthermore, the parallel algorithm we propose has no synchronization overhead, i.e., they only require read-write atomicity.
 
This paper is organized as follows. Section \ref{sec:back} gives background on the LLP method.
Section \ref{sec:prog} gives the programming notation used to express LLP algorithms in the paper.
Section \ref{sec:longest} applies the LLP method to the longest subsequence problem.
Section \ref{sec:optimalBST} give a parallel  algorithm for the optimal binary search tree construction problem.
Section \ref{sec:knapsack}  gives an LLP algorithm for the knapsack problem. 

\section{Background}
\label{sec:back} 
In this section, we cover the background information on the LLP Algorithm \cite{DBLP:conf/spaa/Garg20}.
Let $\CC$ be the lattice of all $n$-dimensional vectors of reals greater than or equal to zero vector and less than or equal to a given vector $T$
where the order on the vectors is defined by the component-wise natural $\leq$.
The lattice is used to model the search space of the combinatorial optimization problem.
The combinatorial optimization problem is modeled as finding the minimum element in $\CC$ that satisfies a boolean {\em predicate} $B$, where
$B$ models {\em feasible} (or acceptable solutions).
We are interested in parallel algorithms to solve the combinatorial optimization problem with $n$ processes.
We will assume that the systems maintains as its state the current candidate vector $\C \in \CC$ in the search lattice, 
where $\C[i]$ is maintained at process $i$. We call $\C$, the global state, and $\C[i]$, the state of process $i$.

Fig. \ref{fig:poset-lattice} shows a finite poset corresponding to $n$ processes ($n$ equals two in the figure), and the corresponding lattice of all eleven global states.

\begin{figure}[htbp]
\begin{center}
\includegraphics[width=2.5in,height=1.0in]{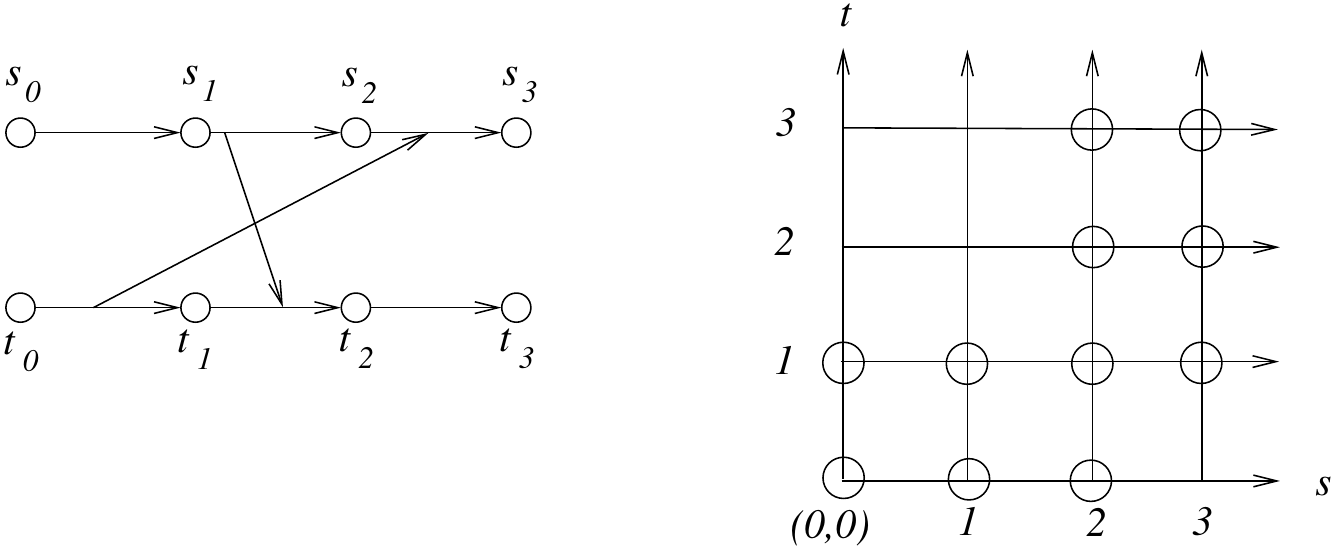}
\caption{A poset and its corresponding distributive lattice $L$ \label{fig:poset-lattice}}
\end{center}
\end{figure}

Finding an element in lattice that satisfies the given predicate $B$, is called the {\em predicate detection} problem.
Finding the {\em minimum} element that satisfies $B$ (whenever it exists) is the combinatorial optimization problem.
A key concept in deriving an efficient  predicate detection algorithm is that of a {\em
forbidden} state.  
Given a predicate $B$, and a vector $\C \in \CC$, a state $G[j]$ is {\em forbidden} (or equivalently, the index $j$ is forbidden) if 
for any vector $H \in \CC$ , where $G \leq H$, if $H[j]$ equals $\C[j]$, then $B$ is false for $H$.
Formally,
\begin{definition}[Forbidden State  \cite{chase1998detection}]
  Given any distributive lattice $\CC$ of  $n$-dimensional vectors of $\RR_{\ge 0}$, and a predicate $B$, we define
  $ \forbidden(G,j,B) \equiv \forall H \in \CC : G \leq H : (G[j] = H[j]) \Rightarrow
  \neg B(H).$
\end{definition}

We define a predicate $B$ to
be {\em lattice-linear} with respect to a lattice $\CC$
 if for any global state $G$,  $B$ is false in $G$ implies that $G$ contains a
{\em forbidden state}. Formally,
\begin{definition}[lattice-linear Predicate  \cite{chase1998detection}]
A boolean predicate $B$ is {\em {lattice-linear}} with respect to a lattice $\CC$
iff
$\forall G \in \CC: \neg B(G) \Ra (\exists j: \forbidden(G,j,B))$.
\end{definition}

Once we determine $j$ such that $forbidden(G,j,B)$, 
we also need to determine how to advance along index $j$.
To that end, we extend the definition of forbidden as follows.
\begin{definition}[$\alpha$-forbidden]
 Let $B$ be any boolean predicate on the lattice $\CC$ of all assignment vectors.
 For any $\C$, $j$ and positive real $\alpha > \C[j]$, we define $\mbox{forbidden}(\C,j, B, \alpha)$ iff
 $$  \forall H \in \CC:H \geq \C: (H[j] < \alpha) \Ra \neg B(H).  $$
\end{definition}

Given any lattice-linear predicate $B$, suppose $\neg B(\C)$. This means that $\C$ must be advanced on all
indices $j$ such that $\forbidden(\C,j,B)$.  We use a function $\alpha(\C,j, B)$ such that $\forbidden(\C, j, B, \alpha(\C,j, B))$ holds
whenever $\forbidden(\C,j,B)$ is true.  With the notion of $\alpha(\C, j, B)$, we have the Algorithm $LLP$.
The algorithm $LLP$ has two inputs --- the predicate $B$ and the top element of the lattice $T$. It returns the least vector $\C$ which is less than or equal to $T$
and satisfies $B$ (if it exists). Whenever $B$ is not true in the current vector $\C$, the algorithm advances on all forbidden indices $j$
in parallel. This simple parallel algorithm can be used to solve a large variety of combinatorial optimization problems
by instantiating different $\forbidden(\C,j,B)$ and $\alpha(\C,j,B)$.

\begin{algorithm}
\SetAlgoRefName{LLP}
 vector {\bf function} getLeastFeasible($T$: vector, $B$: predicate)\\
 {\bf var} $\C$: vector of reals initially $\forall i: \C[i] = 0$;\\
 \h {\bf while} $\exists j: \forbidden(\C,j,B)$ {\bf do}\\
    \h\h    {\bf for all} $j$ such that $\forbidden(\C,j,B)$  {\bf in parallel}:\\
    \h\h\h {\bf if} $(\alpha(\C,j,B) > T[j])$ then return null; \\
 \h   \h\h    {\bf else} $\C[j] := \alpha(\C,j,B)$;\\
      \h {\bf endwhile};\\
    \h {\bf return} $\C$; // the optimal solution
\caption{To find the minimum vector at most $T$ that satisfies $B$\label{fig:alg-llp}}
\end{algorithm}

The following Lemma is useful in proving lattice-linearity of predicates.
\begin{lemma}\label{lem:basic-LLP}  \cite{DBLP:conf/spaa/Garg20,chase1998detection}
Let $B$ be any boolean predicate defined on a lattice $\CC$ of vectors. \\
(a) Let $f:\CC \ra \RR_{\ge 0}$ be any monotone function  defined on the lattice $\CC$ of vectors of $\RR_{\ge 0}$.
Consider the predicate
$B \equiv \C[i] \geq f(\C)$ for some fixed $i$. Then, $B$ is lattice-linear.\\
(b) If $B_1$ and $B_2$ are lattice-linear then $B_1 \wedge B_2$ is also lattice-linear.
\end{lemma}

We now give an example of lattice-linear predicates for scheduling of $n$ jobs. Each job $j$ requires time $t_j$ for completion and has a set of
 prerequisite jobs, denoted by $pre(j)$, such that it can be started only after all its prerequisite jobs
have been completed. Our goal is to find the minimum completion time for each job.
We let our lattice $\CC$ be the set of all possible completion times. A completion vector $\C \in \CC$ is feasible iff $B_{jobs}(\C)$ holds where
$B_{jobs}(\C) \equiv \forall j: (\C[j] \geq t_j) \wedge (\forall i \in pre(j): \C[j] \geq \C[i] + t_j)$.
$B_{jobs}$ is lattice-linear because if it is false, then there exists $j$ such that 
either $\C[j] < t_j$ or $\exists i \in pre(j): \C[j] < \C[i]+t_j$. We claim that $\forbidden(\C, j, B_{jobs})$. Indeed, any vector $H \geq \C$ cannot
be feasible with $\C[j]$ equal to $H[j]$. The minimum of all vectors that satisfy feasibility corresponds to the minimum completion time.

As an example of a predicate that is not lattice-linear, consider the predicate $B \equiv \sum_j \C[j] \geq 1$ defined on the space of 
two dimensional vectors. Consider the vector $\C$ equal to $(0,0)$. The vector $\C$ does not satisfy $B$. For $B$ to be lattice-linear
either the first index or the second index should be forbidden. However, 
none of the indices are
forbidden in $(0,0)$. The index $0$ is not
forbidden because the vector $H = (0,1)$ is greater than $\C$, has $H[0]$ equal to $\C[0]$ but it still satisfies $B$.
The index $1$ is also not forbidden
because $H =(1,0)$ is greater than $\C$, has $H[1]$ equal to $\C[1]$ but it satisfies $B$.

 \section{Notation}
 \label{sec:prog}
 
 We now go over the notation used in description of our parallel algorithms.
Fig. \ref{fig:examples} shows a parallel algorithm for the job-scheduling problems.

The {\bf var} section gives the variables of the problem.
We have a single variable $G$ in the example shown in Fig. \ref{fig:examples}. 
 $G$ is an array of objects such that
 $G[j]$ is the state of thread $j$ for a parallel program.

The {\bf input} section gives all the inputs to the problem. These inputs are constant in the program and do not change during execution.

 The {\bf init} section is used to initialize the state of the program.
 All the parts of the program
 are applicable to all values of $j$. For example, the {\em init} section of the job scheduling program in Fig. \ref{fig:examples}
 specifies that $G[j]$ is initially $t[j]$. Every thread $j$ would initialize $G[j]$.

 The {\bf always} section defines additional variables which are derived from $G$.
 The actual implementation of these variables are left to the system. They can be viewed as
 macros. We will show its use later.

 The LLP algorithm gives the desirable predicate either by using the {\bf forbidden} predicate or {\bf ensure} predicate.
 The {\em forbidden} predicate has an associated {\em advance} clause that specifies how $G[j]$ must be advanced
 whenever the forbidden predicate is true.
 For many problems, it is more convenient to use the complement of the forbidden predicate.
 The {\em ensure} section specifies the desirable predicates of the form $(G[j] \geq expr)$ or 
 $(G[j] \leq expr)$. 
 The statement {\em ensure} $G[j] \geq expr$  simply means that whenever thread $j$ finds $G[j]$ to be less than 
 $expr$; it can advance $G[j]$ to $expr$.
 %
  Since $expr$ may refer to $G$, just by setting $G[j]$ equal to $expr$, there is no guarantee 
  that $G[j]$ continues to be  equal to $expr$ --- the value of $expr$ may change because of changes in other components.
  We use {\em ensure} statement whenever $expr$ is a monotonic function of $G$ and therefore the predicate
  is lattice-linear. 

\begin{figure}
\begin{center}
\small {
\fbox{\begin{minipage}  {\textwidth}\sf
\begin{tabbing}
xxxx\=xxxx\=xxxx\=xxxx\=xxxx\=xxxx\= \kill
$P_j$: Code for thread $j$\\
// common declaration for all the programs below \\
{\bf var} $G$: array[$1$..$n$] of $0..maxint$;// shared among all threads\\
{\bf input}: $t[j]: int$, $pre(j)$: list of $1..n$;\\
{\bf init}: $G[j] := t[j]$;\\
\\
{\bf \underline{job-scheduling}}:\\
\> {\bf forbidden}: $G[j] < \max \{G[i] + t[j] ~|~ i \in pre(j)\}$;\\
\> \> {\bf advance}: $G[j] := \max \{G[i] + t[j] ~|~ i \in pre(j)\}$;\\
\\
{\bf \underline{job-scheduling}}:\\
\> {\bf ensure}: $G[j] \geq \max \{G[i] + t[j] ~|~ i \in pre(j)\}$;\\
\end{tabbing}
\end{minipage}
 }
 }
\caption{LLP Parallel Program for (a) job scheduling problem using forbidden predicate (b) job scheduling problem using ensure clause \label{fig:examples}}
\end{center}
\end{figure}

\section{Longest Increasing Subsequences}
 \label{sec:longest}
 
We are given an integer array as input. For simplicity, we assume that all entries are distinct.
Our goal is find for each index $i$ the length of the longest increasing sequence
that ends at $i$.  For example, suppose the array $A$ is {\tt 
\{35 38 27 45 32\}}. 
Then, the desired output is {\tt 
\{1 2 1 3 2\}}. The corresponding longest increasing subsequences are: {\tt (35), (35, 38), (27), (35, 38, 45), (27, 32).}

We can define a graph $H$ with indices as vertices. For this example, we have five vertices numbered
$v_1$ to $v_5$.
We draw an edge from $v_i$ to $v_j$ if $i$ is less than $j$ and $A[i]$ is also less than $A[j]$.
This graph is clearly acyclic as an edge can only go from a lower index to a higher index.
We use $pre(j)$ to be the set of indices which have an incoming edge to $j$.
The length of the longest increasing subsequence ending at index $j$ is identical to the
length of the longest path ending at $j$. 

To solve the problem using LLP, we model it as a search for the smallest vector $G$ that
satisfies the constraint $B \equiv \forall j:G[j] \geq 1 \wedge \forall j: G[j] \geq \max \{G[i] + 1~|~ i \in pre(j) \}$.
To understand $B$, we first consider a stronger predicate $B_* = (G[1] = 1) \wedge \forall j: G[j]  = \max \{1, \max \{G[i] + 1~|~ i \in pre(j) \} \}$.
The interpretation of $G[j]$ in $B_*$ is that it is the length of the longest path that ends in $j$.
Thus, in the longest increasing subsequence problem we are searching for the vector that satisfies the predicate $B_*$.
Instead of searching for an element in the lattice that satisfies $B_*$, we search for the least element in the lattice that satisfies $B$.
This allows us to solve for the constrained version of the problem in which we are searching for an element that satisfies
an additional lattice-linear constraint.

The underlying lattice we consider is that of all vectors of natural numbers less than or equal to the maximum element
in the lattice. A vector in this lattice is {\em feasible} if it satisfies $B$. 
We first show that the constraint $B$ is lattice-linear.

\noindent
\begin{lemma}
The constraint $B \equiv (\forall j:G[j] \geq 1) \wedge (\forall j: G[j] \geq \max \{G[i] + 1~|~ i \in pre(j) \})$ is lattice-linear.
\end{lemma}
\begin{proof}
Since the predicate $B$ is a conjunction of two predicates, it is sufficient to show that each of them is lattice-linear from
Lemma \ref{lem:basic-LLP}(b). The first conjunct is lattice linear because the constant function $1$ is monotone. The second conjunct
can be viewed as a conjunction over all $j$. For a fixed $j$, the predicate $G[j] \geq \max \{G[i] + 1~|~ i \in pre(j)$ is lattice-linear
from Lemma \ref{lem:basic-LLP}(a).
\end{proof}

Our goal is to find the smallest vector in the lattice that satisfies $B$.
Now, LLP algorithm can be formulated as \ref{fig:alg-longest}.

\begin{algorithm}
\SetAlgoRefName{LLP-Longest-Increasing-Subsequence}
$P_j$: Code for thread $j$\\
 {\bf input}: $A$:array of int;\\
  {\bf var} $G$: array[$1$ \ldots $n$] of int;\\
 {\bf init}: $G[j] = 1$;\\
 \h $pre(j) := \{ i \in 1..j-1 | A[i] < A[j] \}$;\\
 {\bf ensure}: $G[j] \geq \max \{G[i] + 1~|~ i \in pre(j)\}$;\\
 \caption{Finding the Longest Increasing Subsequence. \label{fig:alg-longest}}
\end{algorithm}

This algorithm starts with all values as $1$ and increases the $G$ vector till it satisfies the constraint
$G[j] \geq \max \{G[i] + 1~|~ i \in pre(j) \}$.
The above algorithm, although correct, does not preclude $G[j]$ from getting updated multiple times. 
To ensure that no $G[j]$ is updated more than once, we introduce a boolean $fixed$ for each index such that
we update $G[j]$ only when it is not fixed and all its predecessors are fixed. With this change, our algorithm becomes \ref{fig:alg-longest2}.

\begin{algorithm}
\SetAlgoRefName{LLP2-Longest-Increasing-Subsequence}
$P_j$: Code for thread $j$\\
 {\bf input}: $A$:array of int;\\
 {\bf var} $G$: array[$1$ \ldots $n$] of int;\\
 \h $fixed$:  array[$1$ \ldots $n$] of boolean;\\
 {\bf init}: $G[j] = 1; fixed[j] := false;$ \\
 \h $pre(j) := \{ i \in 1..j-1 | A[i] < A[j] \}$;\\
 {\bf forbidden}: $\neg fixed[j] \wedge (\forall i \in pre(j): fixed[i] )$;\\
\h  {\bf advance}: $G[j] := \max \{G[i] + 1~|~ i \in pre(j)\}$;\\
\h \h \h $fixed[j] := true$;\\
 \caption{Finding the Longest Increasing Subsequence. \label{fig:alg-longest2}}
\end{algorithm}

Let us now analyze the complexity of the algorithm. The sequential complexity is simple because we can maintain the list 
of all vertices that are forbidden because all its predecessors are fixed. Once we have processed a vertex, we never process it again.
This is similar to a sequential algorithm of topological sort.
In this case, we examine a vertex exactly once only after all its predecessors are fixed. The time complexity of this algorithm is
$O(n^2)$.

For the parallel time complexity, assume that we have $n^2$ processors available. Then, in time $O(\log n)$, one can determine whether the
vertex is forbidden and advance it to the correct value if it is forbidden. This is because for every $j$, we simply need to check that all vertices 
in $pre(j)$ are fixed and $j$ is not fixed. By using a parallel {\em reduce} operation, we can check in $O(\log n)$ time whether $j$ is forbidden.
If the longest path in the graph $H$ is $\Delta$, then the
algorithm takes $O(\Delta \log n)$ time. 

Now, let us consider the situation where each thread $j$ writes the value of $fixed[j]$ and $G[j]$ without using any synchronization.
If any thread $j$ reads the old value of $fixed[i]$ for some $i$ in $pre(j)$, it will not update $fixed[j]$ at that point. Eventually, it will read the correct value of $fixed[i]$,
and perform $advance$. We do assume in this version that if a process reads $fixed[i]$ as true, then it reads the correct value of $G[i]$, because $fixed[i]$ is updated
after $G[i]$. 
Consequently, we get the following result.
\begin{lemma}
There exists a parallel algorithm for the longest increasing subsequence problem which uses just read-write atomicity and solves the problem in $O(\Delta \log n)$ time.
\end{lemma}

 We now add lattice-linear constraints to the problem. Instead of the longest increasing subsequence, we may be interested in the longest increasing subsequence that satisfies an additional predicate.
 \begin{lemma}
 All the following predicates are lattice linear.
 \begin{enumerate}
 \item
For any $j$, $G[j]$ is greater than or equal to the longest increasing subsequence of odd integers ending at $j$.
\item
$G[j]$ is greater than or equal to the longest increasing subsequence such that $j^{th}$ element in the subsequence exceeds $(j-1)^{th}$ element by at least $k$.
 \end{enumerate}
 \end{lemma}
 \begin{proof}
 \begin{enumerate}
 \item 
 Since lattice-linear predicates are closed under conjunction, it is sufficient to focus on a fixed $j$. If $G[j]$ is less than the length of the longest increasing subsequence of odd integers ending at $j$, then the index $j$ is forbidden. Unless $j$ is increased the predicate can never become true.
 \item
 We view this predicate as redrawing the directed graph $H$ such that we draw an edge from $v_i$ to $v_j$ if $i$ is less than $j$ and $A[i]+k$ is less than or equal to $A[j]$.
 \end{enumerate}
 \end{proof}
%
%




We note here that the problem can also be solved in parallel using repeated squaring of an appropriate matrix. We do not discuss that method here since it is not work-optimal and generally not efficient in practice.

\section{Optimal Binary Search Tree}
 \label{sec:optimalBST}
Suppose that we have a fixed set of $n$ symbols called {\em keys} with some associated information called {\em values}. Our goal is to build a dictionary based
on binary search tree out of these symbols.
The dictionary supports a single operation search which returns the value associated with the the given key.
We are also given the frequency of each symbol as the argument for the search query. 
The cost of any search for a given key is given by the length of the path from the root of the binary search tree to the node containing 
that key. Given any binary search tree, we can compute the total cost of the tree for all searches.
We would like to build the binary search tree with the least cost.

Let the frequency of key $i$ being searched is $p_i$. We assume that the keys are sorted in increasing order of $p_i$.
Our algorithm is based on building progressively bigger binary search trees.
The main idea is as follows.
 Suppose symbol $k$ is the root of an optimal binary search tree for symbols in the range $[i..j]$.
The root $k$ divides the range into three parts -- the range of indices strictly less than $k$, the index $k$, and the range of indices
strictly greater than $k$. The left or the right range may be empty.
Then, the left subtree and the right subtree must themselves be optimal for their respective ranges.
Let $\C[i,j]$  denote the least cost of any binary search tree
built from symbols in the range $i..j$. 
We use the symbol $s(i,j)$ as the sum of all frequencies from the symbol $i$ to $j$, i.e.,
\[ s(i,j) = \sum_{k=i}^{j} p_k \]
For convenience, we let $s(i,j)$ equal to $0$ whenever $i > j$, i.e., the range is empty.

We now define a lattice linear constraint on $\C[i,j]$.
Let $i \leq k < j$. Consider the cost of the optimal tree such that symbol $k$ is at the root.
The cost has three components: the cost of the left subtree if any, the cost of the search ending at this node itself and the cost of search in the 
right subtree. The cost of the left subtree is 
\[ \C[i,k-1] + s(i,k-1) \]
whenever $i<k$. The cost of the node itself is $s(k,k)$. The cost of the right subtree is
\[ \C[k+1,j] + s(k+1,j) \]
Combining these expressions, we get

\[ \C[i,j] = \min_{i \leq k < j} (\C[i,k-1] + s(i,j) + \C[k+1,j]) \]

This is also the least value of $\C[i,j]$ such that 
\[ \C[i,j] \geq \min_{i \leq k < j} (\C[i,k-1] + s(i,j) + \C[k+1,j]) \]

We now show that the above predicate is lattice-linear.
\begin{lemma}
The constraint $B \equiv \forall i,j: \C[i,j] \geq \min_{i \leq k < j} (\C[i,k-1] + s(i,j) + \C[k+1,j])$ is lattice-linear.
\end{lemma}
\begin{proof}
Suppose that $B$ is false, i.e., $\exists i,j: \C[i,j]  < \min_{i \leq k < j} (\C[i,k-1] + s(i,j) + \C[k+1,j]).$
This means that there exists $i,j,k$ with $i \leq k < j$ such that  $\C[i,j]  <  (\C[i,k-1] + s(i,j) + \C[k+1,j])$.
This means the the index $(i,j)$ is forbidden and unless $\C[i,j]$ is increased, the predicate $B$ can never become true
irrespective of how other components of $\C$ are increased.
\end{proof}

We now have our LLP-based algorithm for Optimal Binary Search Tree as Algorithm \ref{algo:LLP-BST}.
The program has a single variable $G$. It is initialized so that $G[i,i]$ equals $p[i]$ and $G[i,j]$ equals zero whenever $i$ is not equal to $j$. The algorithm advances $G[i,j]$ whenever it is smaller than
$\min_{i \leq k < j} \C[i,k-1] + s(i,j) + \C[k+1,j]$. 
In Algorithm \ref{algo:LLP-BST}, we have used the {\bf awlways} clause as a macro that uses
$s(i,j)$ as a short form for $ \sum_{k=i}^{j} p[k]$.

\begin{algorithm}
\SetAlgoRefName{LLP-OptimalBinarySearchTree}
$P_{i,j}$: Code for thread $(i,j)$\\
 {\bf input}: $p$:array of real;// frequency of each symbol\\
 {\bf init}: $G[i,j] = 0 ~ \forall i \neq j$;\\
\h $G[i,i] = p[i]$;\\
{\bf always}: 
   $s(i,j) = \sum_{k=i}^{j} p[k]$\\
{\bf ensure}:\\
\h  $ \C[i,j] \geq  \min_{i \leq k < j} \C[i,k-1] + s(i,j) + \C[k+1,j]$\\
{\bf priority}: $(j-i)$\\
\caption{Finding An Optimal Binary Search Tree \label{algo:LLP-BST}}
\end{algorithm}


Although, the above algorithm will give us correct answers, it is not efficient as it may update $\C[i,j]$ before $\C[i,k]$ and $\C[k,j]$ for $i\leq k < j$ have stabilized.
However, the following scheduling strategy ensures that we update $\C[i,j]$ at most once. We check for whether $\C[i,j]$ is forbidden in the order
of $j-i$. Hence, initially all $\C[i,j]$ such that $j=i+1$ are updated. This is followed by all $\C[i,j]$ such that $j=i+2$, and so on.
We capture this scheduling strategy with the ${\bf priority}$ statement. We pick $G[i,j]$ to update such that
$(j-i)$ have minimal values.
Of course, our goal is to compute $\C[1,n]$.
With the above strategy of updating $\C[i,j]$, we get that $\C[i,j]$ is updated at most once. Since there are $O(n^2)$ possible values of $\C[i,j]$ and each takes
$O(n)$ work to update, we get the work complexity of $O(n^3)$. On a CREW PRAM, we can compute all $i,j$ with the fixed difference in parallel. By using
$O(\log n)$ span algorithm to compute $\min$, we get the parallel time complexity as $O(n \log n)$.
Thus, we have the following result.
\begin{lemma}
There exists a parallel algorithm for the optimal binary search tree problem which uses just read-write atomicity and solves it in $O(n \log n)$ parallel time.
\end{lemma}

We now consider the constrained versions of the problem.
\begin{lemma}
All the following predicates are lattice linear.
\begin{enumerate}
    \item Key $x$ is not a parent for any key.
    \item The difference in the sizes of the left subtree and the right subtree is at most $1$.
\end{enumerate}
\end{lemma}
\begin{proof}
\begin{enumerate}
    \item This requirement changes the ensure predicate to
     $ \C[i,j] \geq  \min_{i \leq k < j, k \neq x} \C[i,k-1] + s(i,j) + \C[k+1,j]$.
     The right hand side of the constraint continues to be monotonic and therefore it is lattice linear.
     \item This requirement changes the ensure predicate to
     $ \C[i,j] \geq  \min_{i \leq k < j, |k-1-i, j-k-1| \leq 1} \C[i,k-1] + s(i,j) + \C[k+1,j]$.
     This change simply restricts the values of $k$, and the right hand side continues to be monotonic.
\end{enumerate}
\end{proof}

{\em Remark:}
A problem very similar to the optimal Binary Search tree problem is that of constructing an optimal way of multiplying a chain of matrices. Since matrix multiplication is associative,
the product of matrices $(M_1*M_2)*M_3$ is equal to $M_1*(M_2*M_3)$. However, depending upon the dimensions of the matrices, the computational effort may be different.
We let the dimension of matrix $M_i$ be $m_{i-1} \times m_i$. Note that this keep the matrix product well-defined because the dimension of matrix $M_{i+1}$ would be 
$m_i \times m_{i+1}$ and the product $M_i \times M_{i+1}$ is well-defined. 
We can view any evaluation of a chain as a binary tree where the intermediate notes are the multiplication operation and the leaves are the matrices themselves.
Suppose, our goal is to compute the optimal binary tree for multiplying matrices in the range $M_i \ldots M_j$.
Borrowing ideas from the previous section, we let $\C[i,j]$ denote the optimal cost of computing the product of matrices in the range $M_i \ldots M_j$.
Suppose that this product is broken into products of $M_i \ldots M_k$ and $M_{k+1} \ldots M_j$ and then multiplication of these two matrices.
We can compute the cost of this tree as 
\[  \C[i,k] + \C[k+1,j] + m_{i-1}m_{k}m_j \]
Then, we have the following predicate on $\C$.
\[ \C[i,j] \geq \min_{i \leq k < j} (\C[i,k] +  m_{i-1}m_{k}m_j  + \C[k+1,j] \]

The reader will notice the similarity with the optimal binary search tree problem and this problem and the same algorithm can be adapted to solve this problem.

\section{Knapsack Problem}
 \label{sec:knapsack} 
 
We are given $n$ items with weights $w_1, w_2, \ldots, w_n$ and values $v_1, v_2, \ldots, v_n$. We are also given a knapsack
that has a capacity of $W$. Our goal is to determine the subset of items that can be carried in the knapsack and that maximizes the total value.
The standard dynamic programming solution is based on memoization of the following dynamic programming formulation \cite{Vazirani:2001, williamson}.
Let $G[i, w]$ be the maximum value that can be obtained by picking items from $1..i$ with the capacity constraint of $w$.
Then, $G[i,w]= max (G[i-1, w-w_i] + v_i, G[i-1, w])$. The first argument of the max function corresponds to the case when the item $i$ is included in the optimal set
from $1..i$, and the second argument corresponds to the case when the item $i$ is not included and hence the entire capacity can be used for the items from $1..i-1$.
If $w_i > w$, then the item $i$ can never be in the knapsack and can be skipped.
The base cases are simple. The value of $G[0,w]$ and $G[i,0]$ is zero for all $w$ and $i$.
Our goal is to find $G[n,W]$. By filling up the two dimensional array $G$ for all values of $0 \leq i \leq n$ and $0 \leq w \leq W$, we get an algorithm
with time complexity $O(nW)$.

We can model this problem using lattice-linear predicates as follows. 
We model the feasibility as 
$G[i,w] \geq \max (G[i-1, w-w_i] + v_i, G[i-1, w])$ for all $i,w > 0$ and $w_i \leq w$.
Also, 
$G[i,w] =0$ if $i=0$ or $w=0$. Our goal is to find the minimum vector $G$ that satisfies feasibility.
\begin{lemma}
The constraint $B \equiv \forall i,w: G[i,w] \geq \max (G[i-1, w-w_i] + v_i, G[i-1, w])$ for $w_i \leq w$ is lattice-linear.
\end{lemma}
\begin{proof}
If the predicate $B$ is false, there exists $i$ and $w$ such that $G[i,w] < \max (G[i-1, w-w_i] + v_i, G[i-1, w])$. The value $G[i,w]$ is forbidden; unless
$G[i,w]$ is increased the predicate can never become true.
\end{proof}

\begin{algorithm}
\SetAlgoRefName{LLP-Knapsack}
$P_{i,j}$: Code for thread $(i,j)$\\
 {\bf input}: $w,v$:array[$1$..$n$] of int;// weight and value of each item\\
 {\bf var}: $G$:array[$0 \ldots n$, $0 \ldots W$] of int; \\

 {\bf init}: $G[i,j] = 0 ~ if ~ (i=0) \vee (j=0)$;\\
{\bf ensure}:\\
$ G[i,j] \geq  \max \{ G[i -1, j -w_i] + v_i,  G[i-1, j] \} $ ~ if ~$j \geq w_i$\\
\h   \h $ \geq G[i-1, j]$, otherwise.\\
\caption{Finding An Optimal Solution to the Knapsack Problem\label{algo:LLP-Knapsack}}
\end{algorithm}

Algorithm \ref{algo:LLP-Knapsack} updates the value of $G[i,j]$ based only on the values of $G[i-1,.]$. Furthermore, $G[i,j]$ is always at least 
$G[i-1, j]$. Based on this observation, we can simplify the algorithm as follows.
We consider the problem of adding just one item to the knapsack given the constraint that the total weight does not exceed $W$.
We maintain the list of all optimal configurations for each weight less than $W$.

\begin{algorithm}
\SetAlgoRefName{LLP-IncrKnapsack2}
$P_{j}$: Code for thread $j$\\
 {\bf input}: $w,v$: int;// weight and value of the next item\\
 \h $C$: array[$0 \ldots W$] of int;\\
 {\bf var}: $G$:array[$0 \ldots W$] of int; \\
 {\bf init}: $\forall j: G[j] = C[j]$;\\
{\bf ensure}:\\
\h  $ G[j] \geq  C[j-w] + v $ ~ if ~$j \geq w$ \\ 
\caption{Finding An Optimal Solution to the Incremental Knapsack Problem\label{algo:LLP-IncrKnapsack}}
\end{algorithm}

The incremental algorithm can be implemented in $O(1)$ parallel time using $O(W)$ processors as shown in Fig. \ref{algo:LLP-IncrKnapsack}.
Each processor $j$ can check whether $G[j]$ needs to be advanced.

We can now invoke the incremental Knapsack algorithm as Algorithm \ref{algo:Knapsack2}.
If we had $W$ cores, then computing $G[i,.]$ from $G[i-1,.]$ can be done in $O(1)$ giving us the span of $O(n)$.
\begin{algorithm}
\SetAlgoRefName{Knapsack2}
$P_{j}$: Code for thread $j$\\
{\bf input}: $w,v$:array[$1$..$n$] of int;// weight and value of each item\\
{\bf var}: $G$:array[$0 \ldots W$] of int;\\
{\bf init}: $\forall j: G[j] = 0$;\\
{\bf for} $i := 1$ to $n$ do\\
\h     $G := IncrKnapsack2(w[i], v[i], G)$;\\
\caption{Finding An Optimal Solution to the Knapsack Problem\label{algo:Knapsack2}}
\end{algorithm}


We now add some lattice-linear constraints to the Knapsack problem. In many applications, some items may be related and the constraint
$x_a \Rightarrow x_b$ means that if the item $x_a$ is included in the Knapsack then the item $x_b$ must also be included. Thus, the item $x_a$ has profit of zero if $x_b$ is not included.
The item $x_b$ has utility even without $x_a$ but not vice-versa.
Without loss of generality, we assume that all weights are strictly positive, and
that index $b < a$.
In the following Lemma, we use an auxiliary variable $S[i,j]$ that keeps the 
set of items included in $G[i,j]$ and not just the profit from those items.

\begin{lemma}
First assume that $(i \neq a)$.
Let $B(i,w) \equiv G[i,w] \geq \max (G[i-1, w - w_i] + v_i, G[i-1, w])$ for $(w_a \leq w)$ and $G[i,w] \geq G[i-1, w]$, otherwise.
This predicate corresponds to any item $i$ different from $a$. The value with a bag of capacity $w$ is always greater than or equal to the choice of picking the item or not picking the item.

Let $B(a,w) \equiv G[a,w] \geq \max (G[a-1, w-w_a] + v_a, G[a-1, w])$ if $b \in S[a-1, w-w_a] \wedge (w_a \leq w)$ and $G[a,w] \geq G[a-1,w]$, otherwise.

Then, $B(i,w)$ is lattice-linear for all $i$ and $w$.
\end{lemma}
\begin{proof}
Suppose that $B(i,w)$ is false for some $i$ and $w$. Unless $G[i,w]$ is increased, it can never become true.
\end{proof}


\section{Conclusions}
In this paper, we have shown that many dynamic programming problems can be solved using a single {\em parallel} Lattice-Linear Predicate algorithm.
In particular, LLP algorithm solves the problem of the longest  increasing subsequence, the optimal binary search tree and the knapsack problem. In addition, it solves the constrained versions
of these problems. The parallel algorithms described in the paper works correctly with read-write atomicity of variables without any use of {\em locks}.

\bibliography{fmaster}

\begin{thebibliography}{10}

\bibitem{DBLP:conf/spaa/Garg20}
Vijay~K. Garg.
\newblock Predicate detection to solve combinatorial optimization problems.
\newblock In Christian Scheideler and Michael Spear, editors, {\em {SPAA} '20:
  32nd {ACM} Symposium on Parallelism in Algorithms and Architectures, Virtual
  Event, USA, July 15-17, 2020}, pages 235--245. {ACM}, 2020.

\bibitem{bellman1952theory}
Richard Bellman.
\newblock On the theory of dynamic programming.
\newblock {\em Proceedings of the National Academy of Sciences of the United
  States of America}, 38(8):716, 1952.

\bibitem{Birk3}
G.~Birkhoff.
\newblock {\em Lattice Theory}.
\newblock Providence, R.I., 1967.
\newblock third edition.

\bibitem{davey}
B.~A. Davey and H.~A. Priestley.
\newblock {\em Introduction to Lattices and Order}.
\newblock Cambridge University Press, Cambridge, UK, 1990.

\bibitem{Gar:2015:bk}
Vijay~K Garg.
\newblock {\em Lattice Theory with Computer Science Applications}.
\newblock Wiley, New York, NY, 2015.

\bibitem{knuth1971optimum}
Donald~E. Knuth.
\newblock Optimum binary search trees.
\newblock {\em Acta informatica}, 1(1):14--25, 1971.

\bibitem{horowitz1974computing}
Ellis Horowitz and Sartaj Sahni.
\newblock Computing partitions with applications to the knapsack problem.
\newblock {\em Journal of the ACM (JACM)}, 21(2):277--292, 1974.

\bibitem{ibarra1975fast}
Oscar~H Ibarra and Chul~E Kim.
\newblock Fast approximation algorithms for the knapsack and sum of subset
  problems.
\newblock {\em Journal of the ACM (JACM)}, 22(4):463--468, 1975.

\bibitem{chase1998detection}
Craig~M Chase and Vijay~K Garg.
\newblock Detection of global predicates: Techniques and their limitations.
\newblock {\em Distributed Computing}, 11(4):191--201, 1998.

\bibitem{Vazirani:2001}
Vijay~V. Vazirani.
\newblock {\em Approximation Algorithms}.
\newblock Springer-Verlag, Berlin, Germany, 2001.

\bibitem{williamson}
David B~Shmoys David P~Williamson.
\newblock {\em The Design of Approximation Algorithms}.
\newblock Cambridge University Press, 2010.

\end{thebibliography}
\end{document}